\documentclass[journal,twoside,web]{ieeecolor}

\usepackage{generic}
\usepackage{amsmath,amssymb,amsfonts,amsthm}
\usepackage{xcolor}

\usepackage{calc}
\usepackage{cases}

\usepackage{enumerate}
\usepackage{algorithm}
\usepackage[noend]{algorithmic}
\usepackage{mathtools}
\usepackage{svg}

\DeclareMathOperator*{\argmax}{arg\,max}

\newcommand{\change}[1]{{#1}}

\newcommand{\remove}[1]{}

\newcommand{\smallconc}[2]{\begin{bsmallmatrix} #1 \\ #2 \end{bsmallmatrix}}

\newcommand{\clco}{\overline{\operatorname{co}}}
\newcommand{\co}{\operatorname{co}}

\newcommand{\R}{\mathbb{R}}

\newcommand{\IR}{\mathbb{IR}}

\newcommand{\calB}{\mathcal{B}}

\newcommand{\calJ}{\mathcal{J}}

\newcommand{\calM}{\mathcal{M}}
\newcommand{\calN}{\mathcal{N}}

\newcommand{\calR}{\mathcal{R}}
\newcommand{\calS}{\mathcal{S}}

\newcommand{\bbN}{\mathbb{N}}

\newcommand{\ul}[1]{\underline{#1}}
\newcommand{\ula}{\ul{a}}

\newcommand{\ol}[1]{\overline{#1}}
\newcommand{\ola}{\ol{a}}

\newcommand{\olM}{\ol{M}}

\newtheorem{thm}{Theorem}
\newtheorem{lemma}{Lemma}

\newtheorem{prop}{Proposition}
\newtheorem{definition}{Definition}

\usepackage{array}
\usepackage{lipsum}
\usepackage[hidelinks]{hyperref}
\usepackage{xspace}

\newtheorem{example}{Example}

\newtheorem{corollary}{Corollary}

\newtheorem{rem}{Remark}
\theoremstyle{definition}

\definecolor{dblue}{rgb}{.098,.243,.424}
\definecolor{dcompb}{RGB}{157,35,0}  %
\title{A Linear Differential Inclusion for Contraction Analysis to Known Trajectories}
\author{Akash~Harapanahalli,~\IEEEmembership{Graduate Student Member,~IEEE,}
        and~Samuel~Coogan,~\IEEEmembership{Senior~Member,~IEEE}%
\thanks{Akash Harapanahalli and Samuel Coogan are with the School of Electrical and Computer Engineering, Georgia Institute of Technology, Atlanta, GA, USA, 30318. \{\texttt{aharapan},\texttt{sam.coogan}\}\texttt{@gatech.edu}}
}

\usepackage{tikz}
\usetikzlibrary{shapes.geometric,calc,shadings}
\usetikzlibrary{decorations.markings}

\begin{document}
\maketitle

\begin{abstract}
Infinitesimal contraction analysis provides exponential convergence rates between arbitrary pairs of trajectories of a system by studying the system's linearization. 
An essentially equivalent viewpoint arises through stability analysis of a linear differential inclusion (LDI) encompassing the incremental behavior of the system.
In this note, we \change{use contraction tools to study the exponential stability} of a system to a particular known trajectory, deriving a new LDI characterizing the error between arbitrary trajectories and this known trajectory. 
As with classical contraction analysis, this new inclusion is constructed via first partial derivatives of the system's vector field, and \change{convergence} rates are obtained with familiar tools: uniform bounding of the logarithmic norm and LMI-based Lyapunov conditions. 
Our LDI is guaranteed to outperform a usual contraction analysis in two special circumstances: i) when the bound on the logarithmic norm arises from an interval overapproximation of the Jacobian matrix, and ii) when the norm considered is the $\ell_1$ norm.
Finally, we demonstrate how the proposed approach strictly improves an existing framework for ellipsoidal reachable set computation.
\end{abstract}


\section{Introduction}

Contraction theory provides powerful tools for analyzing nonlinear systems by studying their linearizations; see~\cite{AminzareSontag:2014, Bullo:2023, DavydovBullo:2024} for recent surveys on the rich history of contraction analysis in dynamical systems. Applications of contraction analysis include: 
analysis and design of systems with inputs \cite{Sontag:2010} and networked systems~\cite{DeLellisRusso:2010,RussoSontag:2012};
incremental stability in systems with Riemannian \cite{LohmillerSlotine:1998} or Finsler structures~\cite{ForniSepulcre:2013}; 
control design using control contraction metrics~\cite{ManchesterSlotine:2017}; 
Lyapunov function design for monotone systems~\cite{Coogan:2019}; 
robustness analysis of implicit neural networks~\cite{JafarpourBullo:2021,JafarpourCoogan:2022}; 
robust stability with non-Euclidean norms~\cite{DavydovJafarpourBullo_NonEucContr:2022};
observer design with Riemannian metrics~\cite{SanfelicePraly_ConvergenceObserversI:2012};
\change{and local exponential stability about known trajectories on manifolds~\cite{WuYiRantzer_Stability:2024}}.

Consider the nonlinear system $\dot{x} = f(x)$ for differentiable $f:\R^n\to\R^n$.
Let $\calJ =\{\frac{\partial f}{\partial x}(x) : x\in \mathbb{R}^n\}$ be the set of all linearizations of the system.
A key result from contraction theory is that if the logarithmic norm $\mu(\cdot)$ (induced by some norm $|\cdot|$ as defined below) of these linearizations is uniformly bounded by a constant $c\in\R$, that is, $\mu(J) \leq c$ for every $J\in\calJ$, then
\begin{align} \label{eq:contr_bound}
    |x(t) - x'(t)| \leq e^{ct} |x(0) - x'(0)|
\end{align}
for any two trajectories $x$ and $x'$ (see, \emph{e.g.},~\cite[Theorem 3.9]{Bullo:2023}).
In Section~\ref{subsec:LDI_stability}, we recall how the bound \eqref{eq:contr_bound} equivalently arises through stability analysis of the linear differential inclusion (LDI) $\dot{\varepsilon} \in \clco(\calJ)\varepsilon$, characterizing the error dynamics (\change{$\varepsilon = x - x'$}) between two arbitrary trajectories of the system \change{(where $\clco$ denotes the closed convex hull)}.

\begin{figure}
    \centering
    \includesvg[width=\columnwidth]{gamma_vis.svg}
    \caption{A visual comparison of existing approaches (\textbf{left}), and the proposed approach (\textbf{right}) in $\R^3$, with $X = [-1,1]^3$, $x'= -[0.5\ 0.5\ 0.5]^T$, and $x=[0.5\ 0.5\ 0.5]^T$.
    \textbf{Left:} using the straight-line path requires $\calJ$ to include every possible Jacobian matrix $Df(x)$ for $x\in X$.
    \textbf{Right:} using the element-wise path requires $\calM$ to include every possible mixed Jacobian matrix $M_{x'}f(x,s)$ for $x\in X$ and $s\in[0,1]^3$ (see Definition~\ref{def:mixed_Jacobian}).
    For a fixed $s\in[0,1]^3$, $M_{x'}f(x,s)$ consists of the following three columns: $(M_{x'}f(x,s))_{:,1} = (Df(-0.5+s_1,-0.5,-0.5))_{:,1}$, $(M_{x'}f(x,s))_{:,2} = (Df(0.5,-0.5+s_2,-0.5))_{:,2}$, $(M_{x'}f(x,s))_{:,3} = (Df(0.5,0.5,-0.5+s_3))_{:,3}$.
    \change{Correspondingly, we need to bound these columns over $X_1\times \{x'_2\} \times \{x'_3\}$ (green), $X_1\times X_2\times\{x'_3\}$ (red), and $X$ (blue).}
    }
    \label{fig:gamma_vis}
\end{figure}

In this note, given a particular known trajectory $x'$, we build a new linear differential inclusion $\dot{\varepsilon}\in\clco(\calM)\varepsilon$, bounding the error dynamics of the system between an arbitrary trajectory $x$ and the fixed trajectory $x'$, which we subsequently use to study \change{exponential stability} to $x'$.
The key novelty of our approach is to consider an element-wise path between $x'$ and $x$, rather than the straight line connecting them (see Figure~\ref{fig:gamma_vis}), which potentially improves \change{logarithmic norm} estimates when $x'$ is known.
The set of matrices $\calM$ used to construct this LDI is inspired by existing results in the interval analysis literature. A result similar to \change{our Corollary~\ref{cor:intervalM} below} was originally proposed in~\cite{Hanson:1968} to find solutions to the system of equations $f(x) = 0$ and further used to define the \emph{mixed centered inclusion function}~\cite[Section 2.4.4]{Jaulin:2001}, which improves upon a class of Jacobian-based inclusion functions that over-approximates the \change{output} of a function using interval bounds of its \change{first-order partial derivatives}.

One of the features of contraction theory is its generality in comparing any two trajectories of the system, rather than comparing to a known fixed trajectory. For instance, if $c<0$ in~\eqref{eq:contr_bound} (strongly contracting), it can be shown that every trajectory will converge to an \emph{a priori} unknown unique equilibrium point~\cite[Theorem 3.9]{Bullo:2023}. 
However, in many \change{computational} applications of contraction theory, the trajectory $x'$ in~\eqref{eq:contr_bound} is fixed to a known trajectory of the system \change{\emph{at runtime or during execution}}, motivating the setting of this work. 
For example, in reachable set computation~\cite{MaidensArcak:2014,MaidensArcak:2014b,ArcakMaidens:2018}, a single trajectory of the system is simulated and a full reachable tube around this trajectory is computed by expanding or contracting a norm ball using an upper bound of the logarithmic norm.
For robustness analysis of implicit neural networks~\cite{JafarpourBullo:2021,JafarpourCoogan:2022}, contraction theory adds robustness by analyzing a contraction condition around nominal data samples.
Finally, contraction has been used to \change{analyze local stability to known trajectories on manifolds~\cite{WuYiRantzer_Stability:2024} and to} design feedback controllers for trajectory tracking~\cite{ManchesterSlotine:2017}, which fix $x'$ to the trajectory to be tracked.

The note is structured as follows.
In Section~\ref{sec:preliminaries}, we recall how the usual contraction bound~\eqref{eq:contr_bound} equivalently arises through stability analysis of an LDI encompassing the incremental dynamics of the system.
In Section~\ref{sec:newLDI}, we build a similar LDI using our proposed mixed Jacobian operator in Definition~\ref{def:mixed_Jacobian}, where we traverse an element-wise path resulting in a different set of matrices bounding the output of a function. 
In Section~\ref{sec:nlsyscontr}, we apply the LDI to study \change{stability} to an \emph{a priori} known trajectory $x'$, where a uniform bound of the logarithmic norm of our new mixed Jacobian matrix set yields \change{an exponential} bound.
Finally, in Section~\ref{sec:reach}, we use interval analysis to bound the mixed Jacobian matrix, culminating in an improved algorithm for computing ellipsoidal reachable sets.

\section{Preliminaries} \label{sec:preliminaries}

\subsection{Notation}
Let $|\cdot|$ denote a norm on $\R^n$. 
For a matrix $A\in\R^{n\times n}$, let $\|A\| = \sup_{x\in\R^n : |x| = 1} |Ax|$ denote the induced norm on $\R^{n\times n}$. 
Let $\mu(A) = \lim_{h\searrow 0} \frac{\|I + hA\| - 1}{h}$ denote the (induced) logarithmic norm, also called the matrix measure.
For $x\in\R^n$, let $|x|_1 = \sum_{i=1}^n|x_i|$ denote the $\ell_1$-norm, and $\mu_1$ be the induced logarithmic norm.
For $x\in\R^n$ and positive definite $P \succ 0$, let $|x|_{2,P^{1/2}} = \sqrt{x^T Px}$ be the $P$-weighted $\ell_2$ norm on $\R^n$, $\mu_{2,P^{1/2}}$ be its induced logarithmic norm, and $\calB_r^P(x') := \{x\in\R^n : |x - x'|_{2,P^{1/2}} \leq r\}$ be the closed ball of radius $r$ about $x'\in \R^n$. 
The following linear matrix inequality (LMI) provides a convex characterization of $\mu_{2,P^{1/2}}$,
\begin{align} \label{eq:ell2_LMI}
    \mu_{2,P^{1/2}}(M) \leq c \iff M^TP + PM \preceq 2cP.
\end{align}
Let $\IR$ denote the set of closed intervals of $\R$, of the form 
$[\ula,\ola]:=\{a\in\R : \ula\leq a\leq\ola\}$ for $\ula,\ola\in\R\cup\{\pm\infty\}$.
An interval matrix $[\calS]\in\IR^{m\times n}$ is a matrix of closed intervals in $\R$, \emph{i.e.}, $[\calS]_{ij} \in \IR$. 
$[\calS]$ is a subset of $\R^{m\times n}$ in the entrywise sense, so $A\in [\calS]$ is $A_{ij}\in[\calS]_{ij}$ for every $i=1,\dots,m$ and $j=1,\dots,n$.
Given a set $\calS\subseteq\R^{m\times n}$, there is a unique smallest interval $[\calS]$ containing $\calS$, \emph{i.e.}, $\bigcap_{[\calS]\in\IR^{m\times n} : \calS \subseteq [\calS]} [\calS]$. This is well defined since the closed intervals on $\R$ are closed to arbitrary intersection, and $[\calS]$ is defined entrywise.

For any map $f$ and any subset $X$, let $f(X) = \{f(x) : x\in X\}$ denote the set-valued image of $f$ over $X$. \change{Let $\co(X)$ denote the convex hull of $X$ and} $\clco(X) = \operatorname{cl}(\co(X))$ denote the closed convex hull of $X$.
Let $D^+$ denote the upper Dini derivative, \emph{i.e.}, $D^+f(t) = \limsup_{h\searrow 0} \frac{f(t+h) - f(t)}{h}$ for a continuous map $f:\R\to\R$.
For a differentiable map $f:\R^n\to\R^m$, let $Df:\R^n\to\R^{m\times n}$ denote the Jacobian map, such that $(Df(x))_{ij} = \frac{\partial f_i}{\partial x_j}(x)$.

\change{
For the ODE $\dot{x} = f(t,x)$, we call $x:J\to\R^n$ for some interval $J\in\IR$ a \emph{trajectory} if $x$ is differentiable and satisfies the ODE for every $t\in\operatorname{int}(J)$ (we do not require $J$ to be the maximal interval of existence).
We say the trajectory $x$ is in $X\subseteq\R^n$ if $x(t)\in X$ for every $t\in J$.
}

\change{Given a subset $X\subseteq\R^n$, we use the following definitions:
\begin{itemize}
    \item The system $\dot{x} = f(t,x)$ is \change{\emph{uniformly incrementally exponentially stable (UIES) at rate $c \in\R$ in $X$}} if for any two trajectories $x$ and $x'$ in $X$ each defined on $[t_0,T]$,
    \begin{align*}
        |x(t) - x'(t)| \leq e^{c(t - t_0)} |x(t_0) - x'(t_0)|,
    \end{align*}
    for every $t\in[t_0,T]$.
    \item Given a known trajectory $x'$ in $X$ defined on $J$, the system $\dot{x} = f(t,x)$ is \change{\emph{uniformly exponentially stable (UES) to $x'$ at rate $c\in\R$ in $X$}} if for any trajectory $x$ in $X$ defined on $[t_0,T]\subseteq J$,
    \begin{align*}
        |x(t) - x'(t)| \leq e^{c(t - t_0)} |x(t_0) - x'(t_0)|,
    \end{align*}
    for every $t\in [t_0,T]$.
\end{itemize}
}
\change{
We note that in the above definitions, $T$ can take the value $+\infty$, $X$ need not be forward invariant, and $c$ need not be negative.
In the literature, UIES has also been referred to as contraction at rate $c$.
}

\subsection{Stability Analysis of LDIs for Contraction Analysis} \label{subsec:LDI_stability}

In this section, we review approaches for stability analysis of LDIs and their connection to contraction analysis. 
A \emph{linear differential inclusion} (LDI) is given by \cite[p.52]{Boyd:1994uq}
\begin{align}
  \label{eq:LDI}
  \dot{x}\in \Omega x
\end{align}
where $\Omega\subseteq \mathbb{R}^{n\times n}$ is a set of matrices. Any $x:[t_0,T]\to\R^n$ satisfying \eqref{eq:LDI} is called a \emph{trajectory} of the LDI. 

Stability of the LDI~\eqref{eq:LDI} has been well studied in the context of robustness analysis of (time-varying) linear systems under uncertainties~\cite{MolchanovPyatnitskiy_StabDIControl:1989,HuBlanchini_LMIStabLDI:2010}. 
In these settings, $\Omega$ is generally an \emph{a priori} known polytope of possible parameters, and the goal is to ensure that every possible choice of $A\in\Omega$ leads to stable system behavior.
The following Lemma recalls a standard result whereby if the logarithmic norm of every matrix $M\in\Omega$ is uniformly bounded by $c$, then any trajectory $x(t)$ of the LDI \eqref{eq:LDI} is norm bounded by a factor of $e^{ct}$.

\begin{lemma}
  \label{lem:LDI_ect}
Consider the LDI $\dot{x} \in \Omega x$ and some norm $|\cdot|$ on $\mathbb{R}^n$. If $\mu(M)\leq c$ for all $M\in \Omega$, then \change{for any trajectory $x$ defined on $[t_0,T]$},
\begin{align*}
    |x(t)|\leq e^{c(t - t_0)}|x(t_0)|,
\end{align*}
for every $t\in[t_0,T]$.
\end{lemma}
\change{
\begin{proof}
    For any $t\in[t_0,T]$ there exists $M(t)\in\Omega$ such that
    \begin{align*}
        D^+& |x(t)| = \textstyle\limsup_{h\searrow 0} \tfrac1h (|x(t + h)| - |x(t)|) \\
        &= \textstyle\limsup_{h\searrow 0} \tfrac1h (|x(t) + hM(t)x(t)| - |x(t)|)  \\
        &\leq \textstyle\limsup_{h\searrow 0} \tfrac1h (\|I + hM(t)\| - 1) |x(t)| \leq c|x(t)|.
    \end{align*}
    The Gr\"onwall inequality~\cite[Prop. 2]{Lorenz_Mutational} completes the proof.
\end{proof}
}
Lemma~\ref{lem:LDI_ect} can be viewed as a corollary of Coppel's inequalities \cite[Theorem 27]{Desoer:2008bh}.
Lemma~\ref{lem:LDI_ect} and the \change{LMI}~\eqref{eq:ell2_LMI} are the essential ingredients for quadratic stability analysis of LDIs \cite[Ch. 4 and 5]{Boyd:1994uq}, where the \change{LMI} is used as a constraint in a semi-definite program.
Next, we recall an inclusion obtained using the mean value theorem and ideas from convex analysis.

\begin{prop}[{\relax\cite[p.55]{Boyd:1994uq}}] \label{prop:J_LDI}
    Let $f:\R^n\to\R^m$ be differentiable and $X\subseteq\R^n$ be convex. If $\calJ\subseteq\R^{m\times n}$ satisfies
    \begin{align*}
        Df(X) \subseteq \calJ,
    \end{align*}
    then 
    \begin{align} \label{eq:J_LDI}
        f(x) - f(x') \in \clco(\calJ) (x - x')
    \end{align}
    for every $x,x'\in X$.
\end{prop}
The proof is in~\cite[Section 4.3.1, p.55]{Boyd:1994uq}; we provide it here for comparison with that of Theorem~\ref{thm:mixed} in the next Section.
\begin{proof}
    Fix $\ell\in\R^m$ and $x,x'\in X$. Consider the curve $\gamma:[0,1]\to X$, $\gamma(s) = sx + (1-s)x'$. Since $\gamma$ is continuous on $[0,1]$ and differentiable on $(0,1)$, applying the mean value theorem, there exists $s'\in(0,1)$ such that
    \begin{align*}
        \ell^T(f(\gamma(1)) - f(\gamma(0))) = \ell^T Df(\gamma(s'))(\gamma(1) - \gamma(0)).
    \end{align*}
    Since $\gamma(s)\in X$ by convexity of $X$, $Df(\gamma(s'))\in\calJ$. Thus,
    \begin{align*}
        \ell^T(f(x) - f(x')) \leq \sup_{J\in\clco(\calJ)} \ell^T J(x - x'),
    \end{align*}
    which implies that $f(x) - f(x')$ belongs to every halfspace containing $\clco(\calJ)(x-x')$, since $\ell$ was arbitrary. 
    But since $\clco(\calJ)(x - x')$ is closed and convex, it equals the intersection of these halfspaces, leading to~\eqref{eq:J_LDI}.
\end{proof}

Lemma~\ref{lem:LDI_ect} and Proposition~\ref{prop:J_LDI} \change{together} recover a standard result from infinitesimal contraction theory. Consider the system
\begin{align*}
    \dot{x} = f(x),
\end{align*}
for $x\in\R^n$, $f:\R^n\to\R^n$ continuously differentiable, and let $X\subseteq\R^n$ be a convex set. If $Df(X)\subseteq\calJ$, then Proposition~\ref{prop:J_LDI} builds the following LDI \change{on $\varepsilon = x - x'$, for any $x,x'\in\R^n$,}
\begin{align*}
    \dot{\varepsilon} = f(x) - f(x') \in \clco(\calJ) (x - x') = \clco(\calJ)\varepsilon.
\end{align*}
Let $\sup_{J\in\calJ} \mu(J) \leq c$. By convexity \change{and continuity} of $\mu$, we have $\sup_{J\in \clco(\calJ)} \mu(J) \leq c$ \change{(see the proof of Theorem~\ref{thm:main} below)}.
Applying Lemma~\ref{lem:LDI_ect}, we therefore see that
\begin{align} \label{eq:basic_contraction}
    \sup_{x\in X}\mu(Df(x)) \leq c \implies |x(t) - x'(t)| \leq e^{ct} |x_0 - x'_0|,
\end{align}
where $x,x'$ are trajectories \change{in $X$} from initial conditions $x_0,x'_0\in X$ at time $0$.
In other words, a uniform bound $c$ for the logarithmic norm $\mu(Df(x))$ implies that the system $\dot{x} = f(x)$ is \change{UIES} at rate $c$. 

\change{An alternate proof of this result integrates along the line segment $\gamma$ as $A(x) = \int_0^1 Df(sx + (1-s)x')\ \mathrm{d} s$ replacing the expression~\eqref{eq:J_LDI} with $f(x) - f(x') = A(x) (x - x')$~\cite[Lemma 1]{Sontag:2010}~\cite[Lemma 2]{AminzareSontag:2014}.
The LDI~\eqref{eq:J_LDI} avoids the need to compute an integral when overapproximating these matrices.
Other proof techniques include verifying a Finsler-Lyapunov condition on the tangent bundle~\cite{ForniSepulcre:2013,WuDuan_GeoLyapCharISSFinsler:2022}, and using weak pairings compatible with the norm~\cite{DavydovJafarpourBullo_NonEucContr:2022}.
}

In the next section, we show how the LDI viewpoint allows for a modification when $x'$ is fixed to a known trajectory.

\section{The Mixed Jacobian Linear Inclusion for Differentiable Mappings} \label{sec:newLDI}

In this section, we build a new linear inclusion that characterizes the behavior of a general differentiable map $f:\R^n\to\R^m$ when comparing an arbitrary $x\in\R^n$ to a fixed $x'\in\R^n$. As with the linear inclusion from Proposition~\ref{prop:J_LDI}, our inclusion is built using first-derivatives of the function $f$.

\subsection{The Mixed Jacobian for the New Linear Inclusion}

In order to define the linear inclusion, we first define a new differential operator constructing a matrix with a particular structure in its partial derivative evaluations.

\begin{definition}[Mixed Jacobian matrix] \label{def:mixed_Jacobian}
Given $x'\in\R^n$, define the \emph{mixed Jacobian operator} $M_{x'}$, such that for differentiable $f:\R^n\to\R^m$, $M_{x'}f : \R^n\times[0,1]^n\to\R^{m\times n}$ where
\begin{align*}
    (&M_{x'}f (x, s))_{ij}  \\
    &= \frac{\partial f_i}{\partial x_j} (x_1,\dots,x_{j-1},s_jx_j + (1-s_j)x'_j,x'_{j+1},\dots,x'_n).
\end{align*}
The matrix $M_{x'}f(x,s)$ is called the \emph{mixed Jacobian matrix} of $f$ at $(x,s)$, since it mixes the inputs to the Jacobian between the point $x'$ and $x$. 
\end{definition}

In the following Theorem, we present the first contribution of this work: a new linear inclusion bounding the behavior of a differentiable map $f$.
As seen in its proof, the set of mixed Jacobian matrices between $x'$ and $x$ characterizes the partial derivatives along an elementwise path between them.

\begin{thm}
  \label{thm:mixed}
  Let $f:\mathbb{R}^n\to\mathbb{R}^m$ be differentiable, 
  $X\subseteq\R^n$, and consider some fixed $x'\in X$.
  If $\calM \subseteq\R^{m\times n}$ satisfies
  \begin{align*}
       M_{x'}f(X,[0,1]^n) \subseteq \calM,
  \end{align*}
  then
  \begin{align}
    \label{eq:M_LDI}
        f(x)-f(x')\in \clco(\calM)(x-x')
  \end{align}
  for every $x\in X$.
\end{thm}

\begin{proof}
    Fix $\ell\in\R^m$ and $x\in X$. For each $k=1,\dots,n$, consider the curve $\gamma_k:[0,1]\to\R^n$,
    \begin{align*}
        \gamma_k(s) = [x_1\cdots x_{k-1}\ sx_k + (1-s)x'_k\ x'_{k+1}\cdots x'_n]^T.
    \end{align*}
    Each curve $\gamma_k$ is continuous on $[0,1]$ and differentiable on $(0,1)$, thus using the mean value theorem there exists $s_k\in (0,1)$ such that
    \begin{align*}
        \ell^T&(f(\gamma_k(1)) - f(\gamma_k(0))) = \ell^T(Df(\gamma(s_k))(\gamma_k(1) - \gamma_k(0))) \\
        &= \textstyle\ell^T\sum_{j=1}^n\frac{\partial f}{\partial x_j}(\gamma_k(s_k))((\gamma_k(1))_j - (\gamma_k(0))_j).
    \end{align*}
    Note that $\gamma_k(1) = [x_1\cdots x_k\ x'_{k+1}\cdots x'_n]^T = \gamma_{k+1}(0)$ for every $k=1,\dots,n-1$. Thus, summing over $k=1,\dots,n$, the LHS is telescoping. Swapping the order of summation on the RHS, we see that
    \begin{align}
        \ell^T&(f(\gamma_n(1)) - f(\gamma_1(0))) = \ell^T(f(x) - f(x')) \nonumber \\ 
        &= \ell^T \textstyle\sum_{j=1}^n \sum_{k=1}^n \frac{\partial f}{\partial x_j}(\gamma_k(s_k))((\gamma_k(1))_j - (\gamma_k(0))_j) \nonumber \\
        &= \textstyle\ell^T\sum_{j=1}^n \frac{\partial f}{\partial x_j}(\gamma_j(s_j))(x_j - x'_j), \label{eq:pf:taylor}
    \end{align}
    where~\eqref{eq:pf:taylor} follows since $(\gamma_k(1))_j - (\gamma_k(0))_j = x_j - x_j = 0$ if $j\leq k-1$, $(\gamma_k(1))_j - (\gamma_k(0))_j = x'_j - x'_j = 0$ if $j\geq k+1$, and $(\gamma_k(1))_j - (\gamma_k(0))_j = x_j - x'_j$ when $j=k$. 
    Finally, since $\gamma_j(s_j) = [x_1\cdots x_{j-1}\ s_jx_j + (1-s_j)x'_j\ x'_{j+1}\cdots x'_n]$,
    \begin{align*}
        \ell^T(f(x) - f(x')) &= \ell^T \textstyle\sum_{j=1}^n (M_{x'}f(x,s))_{:,j}(x_j - x'_j) \\
        &\leq \textstyle\sup_{M\in \clco(\calM)} \ell^T M(x - x'),
    \end{align*}
    since $M_{x'}f(x,s)\in \calM$. 
    Thus, $f(x) - f(x')$ belongs to every halfspace containing $\clco(\calM)(x - x')$, since $\ell$ was arbitrary.
    But since $\clco(\calM)(x - x')$ is closed and convex, it equals the intersection of these halfspaces, leading to~\eqref{eq:M_LDI}.
\end{proof}

The key feature of the proof of Theorem~\ref{thm:mixed} is the particular path constructed from $x'$ to $x$.
Instead of traversing the straight line segment as in Proposition~\ref{prop:J_LDI}, we construct a path which only changes along one coordinate at a time (see Figure~\ref{fig:gamma_vis} for a visualization).
Repeated application of the mean value theorem on these $n$ different segments builds the vector $s\in[0,1]^n$, characterizing $n$ different points along the elementwise path. The mixed Jacobian matrix $M_{x'}f(x,s)$ carries the corresponding Jacobian column at each of these $n$ points.

\subsection{Interval Overapproximations of the (Mixed) Jacobian}

In either Proposition~\ref{prop:J_LDI} or Theorem~\ref{thm:mixed}, a natural question is how to build a matrix set $\calJ$ and $\calM$ satisfying the stated assumptions. 
Analytically, it may be possible to write, in closed form, the true image of the Jacobian operator $Df(X)$ or the mixed Jacobian operator $M_{x'}(X,[0,1]^n)$. 
This approach is used in the next section for Examples~\ref{ex:poly} and~\ref{ex:l1}. 

For automated analysis, a closed form expression for the true images may be difficult to derive.
Instead, a more tractable approach may be to construct a set overapproximating the true image.
For instance, the approach developed in~\cite{FanMitra:2016,FanMitra:2017} uses interval analysis~\cite{Jaulin:2001} to obtain an interval matrix $[\calJ]$ overapproximating the set $Df(X)$. 
Interval analysis propagates interval overapproximations through functional building blocks to automatically bound each entry of the Jacobian matrix into their own intervals as $\frac{\partial f_i}{\partial x_j}(X) \subseteq [\calJ]_{ij}$.

The following Corollary shows how to build an interval matrix $[\calM]$ containing the mixed Jacobian matrices, which coincides with the mixed Jacobian interval matrix from~\cite{Hanson:1968,Jaulin:2001}. 
\change{In the mixed Jacobian case, for an interval set $X$ and $x'\in X$, the line segments along the element-wise path live in strict subsets of $X$. For instance, the entries in the first column of $M_{x'}f(x,s)$ only need to bound on input set $X_1\times \{x'_2\} \times \cdots \times \{x'_n\}$, since the corresponding line segment lives on this subset (see Figure~\ref{fig:gamma_vis}).}
Further,  the smallest interval matrix $[\calM]$ is \change{never larger} than the smallest interval matrix $[\calJ]$ on an interval initial set.%

\begin{corollary}[Interval approximations] \label{cor:intervalM}
    Let $f:\R^n\to\R^m$ be differentiable, $X=X_1\times \cdots \times X_n\in\IR^n$ be an interval, and consider some fixed $x'\in X$. 
    The following statements hold:
    \begin{enumerate}[i)]
        \item An interval matrix $[\calJ]$ satisfies $Df(X)\subseteq[\calJ]$ if for every $i=1,\dots,m$ and $j=1,\dots,n$,
        \begin{align} \label{eq:intervalJcond}
            \frac{\partial f_i}{\partial x_j} (X_1,\dots,X_j,X_{j+1},\dots,X_n) \subseteq [\calJ]_{ij};
        \end{align}
        \item An interval matrix $[\calM]$ satisfies $M_{x'}f(X,[0,1]^n)\subseteq[\calM]$ if for every $i=1,\dots,m$ and $j=1,\dots,n$,
        \begin{align} \label{eq:intervalMcond}
            \frac{\partial f_i}{\partial x_j}(X_1,\dots,X_j,x'_{j+1},\dots,x'_n) \subseteq [\calM]_{ij}.
        \end{align}
    \end{enumerate}
    Moreover, the smallest interval matrices $[\calJ]$ and $[\calM]$ satisfying~\eqref{eq:intervalJcond} and~\eqref{eq:intervalMcond} respectively also satisfy $[\calM] \subseteq [\calJ]$.

\end{corollary}
\begin{proof}
    The statement (i) is clear, as interval matrix inclusion is elementwise. 
    Regarding statement (ii), since $X$ is an interval, any $x\in X$ and $s\in [0,1]^n$ satisfies $[x_1 \cdots x_{j-1}\ (1-s_j)x_j + s_jx'_j\ x'_{j+1} \cdots x'_n]^T \in \change{X_1\times \cdots \times X_j \times \{x'_{j+1}\}\times \cdots\times\{x'_n\}}$ for $j=1,\dots,n$. Thus, $[\calM]$ satisfies $M_{x'}\change{f}(X,[0,1]^n)\subseteq[\calM]$.
    Lastly, \change{any interval matrix $[\calJ]$ satisfying~\eqref{eq:intervalJcond}} clearly satisfies~\eqref{eq:intervalMcond} since $x'_j\in X_j$ for every $j$, so if $[\calM]$ is the smallest interval matrix satisfying~\eqref{eq:intervalMcond}, then $[\calM] \subseteq [\calJ]$.
\end{proof}

\begin{rem}[Connection to interval analysis literature]
    Our bound in Theorem~\ref{thm:mixed} is inspired by known results in the interval analysis literature, and Corollary~\ref{cor:intervalM} is essentially equivalent to the result from~\cite{Hanson:1968}. 
    The focus in~\cite{Hanson:1968} is in finding solutions to the system of equations $f(x) = 0$, rather than analyzing the nonlinear dynamical system $\dot{x} = f(x)$.
    The interval matrix $[\calM]$ has also been used to construct interval inclusion functions in~\cite{Jaulin:2001} for robustness analysis of the map $f$.
\end{rem}

\begin{rem}[Interval overapproximations] \label{rem:interval_over}
    More generally, $[\calM]\subseteq[\calJ]$ in Corollary~\ref{cor:intervalM} whenever an \emph{inclusion monotonic inclusion function}~\cite[Section 2.4]{Jaulin:2001} of $\frac{\partial f_i}{\partial x_j}$ is used to obtain the entrywise bounds from~\eqref{eq:intervalJcond} and~\eqref{eq:intervalMcond}.
    \change{In particular, \texttt{immrax}~\cite{immrax} automatically builds such inclusion functions yielding interval matrices satisfying \eqref{eq:intervalJcond} and \eqref{eq:intervalMcond} through automatic differentiation and composition of inclusion monotonic building blocks.}
\end{rem}

\section{Contraction Analysis to Known Trajectories} \label{sec:nlsyscontr}

In this section, we apply the inclusion from Theorem~\ref{thm:mixed} \change{to study exponential stability} of nonlinear systems to known trajectories \change{using tools from contraction theory}.
Consider the following time-varying nonlinear system
\begin{align}\label{eq:nlsys}
    \dot{x} = f(t,x) = f_t(x), \quad x(t_0) = x_0,
\end{align}
where $f:\R\times \R^n\to\R^n$ is continuous, with continuous partial derivatives with respect to $x$, \emph{i.e.}, $(t,x) \mapsto \frac{\partial f}{\partial x}(t,x)$ is continuous. Let $f_t:\R^n\to\R^n$ denote $x\mapsto f_t(x) = f(t,x)$.

\subsection{Contraction to a Known Trajectory}

Mirroring the analysis from Section~\ref{subsec:LDI_stability}, the following Theorem uses a uniform bound of the logarithmic norm of the mixed Jacobian matrices to construct \change{a UES} bound between arbitrary trajectories $x$ and the particular known trajectory $x'$. 

\begin{thm} \label{thm:main}
    Let $|\cdot|$ be a norm with induced logarithmic norm $\mu$, and let $X\subseteq\R^n$ be a set. Consider the dynamical system $\dot{x} = f_t(x)$ from~\eqref{eq:nlsys}.
    Let $x'$ be a known trajectory \change{in $X$} defined on \change{$J$}. If for some $c\in\R$,
    \begin{align*}
        \sup_{t\in J,\,x\in X,\,s\in[0,1]^n} \mu(M_{x'(t)} f_t (x,s)) \leq c,
    \end{align*}
    then the system is \change{uniformly exponentially stable to $x'$ at rate $c$ in $X$}, \emph{i.e.}, for any trajectory $x$ in $X$ defined on \change{$[t_0,T]\subseteq J$,}
    \begin{align*}
        |x(t) - x'(t)| \leq e^{c(t - t_0)} |x(t_0) - x'(t_0)|,
    \end{align*}
    for every \change{$t\in[t_0,T]$}.
\end{thm}
\begin{proof}
    For every \change{$t\in[t_0,T]$}, set $\calM_t := M_{x'(t)}f_t(X,[0,1]^n)$. Letting $\varepsilon = x - x'$, we observe that for any \change{$t\in[t_0,T]$,}
    \begin{align*}
        \dot{\varepsilon}(t) = \dot{x}(t) - \dot{x}'(t) = f_t(x(t)) - f_t(x'(t)),
    \end{align*}
    since $x$ and $x'$ are trajectories. Applying Theorem~\ref{thm:mixed} to the map $f_t:\R^n\to\R^n$, since $x(t),x'(t)\in X$, for every \change{$t\in[t_0,T]$,}
    \begin{align} \label{eq:pf:epsM_t}
        \dot{\varepsilon}(t) \in \clco(\calM_t) \varepsilon(t).
    \end{align}
    Fix \change{$t\in[t_0,T]$}. Equation~\eqref{eq:pf:epsM_t} implies that for every $\tau\in[t_0,t]$,
    \begin{align*}
        \dot{\varepsilon}(\tau) \in \underbrace{\left(\textstyle\bigcup_{\tau'\in[t_0,t]} \clco(\calM_{\tau'})\right)}_{=:\calN}\varepsilon(\tau).
    \end{align*}
    \textit{Claim:} $\sup_{\olM\in\calN} \mu(\olM) \leq c$. 
    Indeed, fix $\olM\in\calN$; 
    there exists $\tau\in[t_0,t]$ such that $\olM \in \clco(\calM_\tau)$.
    Let $\{M_j\}_{j=1}^\infty$ be a sequence satisfying $M_j\in\co(\calM_\tau)$ and $M_j \to \olM$.
    Since $\mu$ is convex~\cite[Lemma 2.11]{Bullo:2023}, $\mu(M_j) \leq c$ for every $j$ since $\sup_{M\in \calM_\tau}\mu(M) \leq c$ by assumption. Since $\mu$ is continuous, 
    \begin{align*}
        \mu(\olM) = \lim_{j\to\infty} \mu(M_j) \leq \sup_{j\geq 1} \mu(M_j) \leq c.
    \end{align*}
    Applying Lemma~\ref{lem:LDI_ect} to the LDI $\dot{\varepsilon} = \calN\varepsilon$, \change{for every $\tau\in[t_0,t]$,}
    \begin{align*}
        |\varepsilon(\tau)| \leq e^{c(\tau-t_0)} |\varepsilon(t_0)|,
    \end{align*}
    But \change{$t\in[t_0,T]$} was arbitrary, completing the proof.
\end{proof}

The next Corollary considers the special case of a nonlinear system where $x'$ is fixed to an equilibrium.

\change{
\begin{corollary}[Equilibrium exponential stability] \label{cor:exp_stable}
    Consider the dynamical system $\dot{x} = f_t(x)$ from~\eqref{eq:nlsys}. Let $x'\in\R^n$ such that $f_t(x') = 0$ for every $t\geq 0$, $X\subseteq\R^n$ such that $x'\in \operatorname{int}(X)$.
    If there exists $c < 0$ satisfying
    \begin{align*}
        \sup_{t\geq0,x\in X,s\in[0,1]^n} \mu(M_{x'}f_t(x,s)) \leq c < 0,
    \end{align*}
    then $x'$ is a locally exponentially stable equilibrium point for which any norm ball $\{x : |x - x'| \leq \gamma\}$ for $\gamma > 0$ contained in $X$ is a forward invariant region of attraction of $x'$.
\end{corollary}
}
\begin{proof}[Proof]
\change{
    Suppose $x$ is a trajectory in $X$ defined on $[t_0,T]\subseteq[0,\infty]$.
    Calling $x'(t) = x'$ the known trajectory, Theorem~\ref{thm:main} implies $|x(t) - x'| \leq e^{c(t - t_0)}|x(t_0) - x'|$ for $t\in[t_0,T]$.
    Any $\calS_\gamma = \{x : |x - x'| \leq \gamma\} \subseteq X$ for $\gamma>0$ is (i) forward invariant since $x(t_0)\in\calS_\gamma$ implies $|x(t) - x'| \leq e^{c(t - t_0)}|x(t_0) - x'| \leq e^{c(t-t_0)}\gamma \leq \gamma$ for every $t\in[t_0,T]$, (ii) a region of attraction of $x'$ since $|x(t) - x'| \leq e^{c(t - t_0)}\gamma \to 0$ as $t\to \infty$.
    Finally, since $x'\in\operatorname{int}(X)$, there is a $\gamma^*>0$ such that $\calS_{\gamma^*} \subseteq X$.
}
\end{proof}

Corollary~\ref{cor:exp_stable} provides another framing of the contribution of this work: first partial derivatives of the vector field $f$ and tools from contraction analysis verify a norm-based Lyapunov condition for the system, \change{with real regions of attraction given by norm balls contained in the original set of interest}.

\begin{rem}[Comparison to classical contraction]
    Theorem~\ref{thm:main} retains many of the key features of classical contraction analysis, such as the forgetting of initial conditions when strongly contracting ($c<0$). 
    The main drawback of our approach is that $x'$ needs to be known beforehand---considering the set of all possible linearizations $\calJ$ guarantees the existence of $x'$ without knowing it \emph{a priori}.
    When $s = 0$, $M_{x'}f(x,0) = Df(x')$, so the Jacobian at $x'$ is included in the set $\calM$ in Theorem~\ref{thm:main}. 
    Fixing $x'$ is therefore crucial for any benefit from Theorem~\ref{thm:main}, since letting $x'$ vary arbitrarily in $X$ yields the following containment, $M_{X}(X,[0,1]^n) \supseteq Df(X)$.
\end{rem}

\change{
\begin{rem}[Connection to equilibrium contraction]
    \cite[Theorem 33]{DavydovJafarpourBullo_NonEucContr:2022} provides the following sufficient condition for UES to an equilibrium point $x'$: if there is a matrix $A(t,x)$ and $c\in\R$ such that for every $t,x$, $f(t,x) = A(t,x)(x - x')$ and $\mu(A(t,x)) \leq c$, then $x'$ is UES at rate $c$.
    Theorem~\ref{thm:mixed} implies the existence of $A(t,x)\in\clco(M_{x'}f_t(X,[0,1]^n))$ such that $f(t,x) = A(t,x) (x - x')$, and the hypothesis of Corollary~\ref{cor:exp_stable} with convexity of $\mu$ implies that $\mu(A(t,x)) \leq c$.
\end{rem}
}

\begin{rem}[Permuting state variables and varying bases] \label{rem:permuting}
    Reordering the state variables will result in different element-wise paths taken from $x$ to $x'$, generally yielding different mixed Jacobian matrices and potentially different contraction rate guarantees from Theorem~\ref{thm:main}. 
    Generally, for any basis of $\R^n$, a similar path between $x$ and $x'$ is constructed by traversing each basis vector direction individually, which computationally corresponds to the same elementwise path in state transformed coordinates.
\end{rem}

\subsection{Analytical Examples}

In this subsection, we demonstrate the advantage of Theorem~\ref{thm:main} with two analytical examples.

\change{
\begin{example} \label{ex:taninv}
  Consider the time-varying nonlinear system
  \begin{align}
    \label{eq:4}
    \dot{x}=f(t,x)
    =\begin{bmatrix}
        -x_1-\tan^{-1}(x_1)x_2  \change{ + \cos{2\pi t}}\\
        -x_2
    \end{bmatrix}.
  \end{align}
  We will show that the system globally entrains to any known trajectory $x'$ starting from an initial condition satisfying $x'_2(0) = 0$.
  \change{
This implies $x'_2(t) = 0$ for every $t\geq 0$.
We have
\begin{align*}
    Df_t(x) =
    \begin{bmatrix}
      -1-\frac{x_2}{1+x_1^2}&-\tan^{-1}(x_1)\\
      0&-1
    \end{bmatrix}. 
\end{align*}
For any interval $X_1\times X_2\subseteq \mathbb{R}^2$, we have
\begin{align*}
\textstyle\frac{\partial f_1}{\partial x_1}(t,X_1,x'_2(t))&=-1 - \tfrac{0}{1 + x_1^2} = -1,\\
\textstyle\frac{\partial f_1}{\partial x_2}(t,X_1,X_2)&=-\tan^{-1}(X_1) \subseteq [-\pi/2,\pi/2].
\end{align*}
  \change{
  Therefore, for any $t\geq0$ and $x\in\R^n$, $M_{x'(t)} f_t(x,s) \subseteq \co\{M_+,M_-\}$, where $ M_{\pm}= \begin{bsmallmatrix} -1&\pm \frac{\pi}{2}\\ 0&-1 \end{bsmallmatrix}$.
  It is easy to show $M_++M_+^T\preceq 2cI$ and $M_-+M_-^T\preceq 2c I$ for $c= \left(\frac{\pi}{4}-1\right)<0$, proving that for any trajectory $x$, $|x(t) - x'(t)|_2 \leq e^{ct} |x(0) - x'(0)|_2$ for every $t\geq 0$, using Theorem~\ref{thm:main}.
  
 In contrast, since $Df_t(x)$ is triangular and $\frac{\partial f_1}{\partial x_1}(x)=-1-\frac{x_2}{1+x_1^2} \change{= \lambda(x)}$ can be positive, the system cannot be contracting in any norm\change{, since $v=[1\ 0]^T$ results in $\|I + hDf(x)\| \geq |v + hc(x)v| = 1 + h\lambda(x)$, and therefore $\mu(Df(t,x)) \geq \lambda(x)$.}
 }
 }
\end{example}
}

\begin{figure}
    \centering
    \includegraphics[width=0.49\columnwidth]{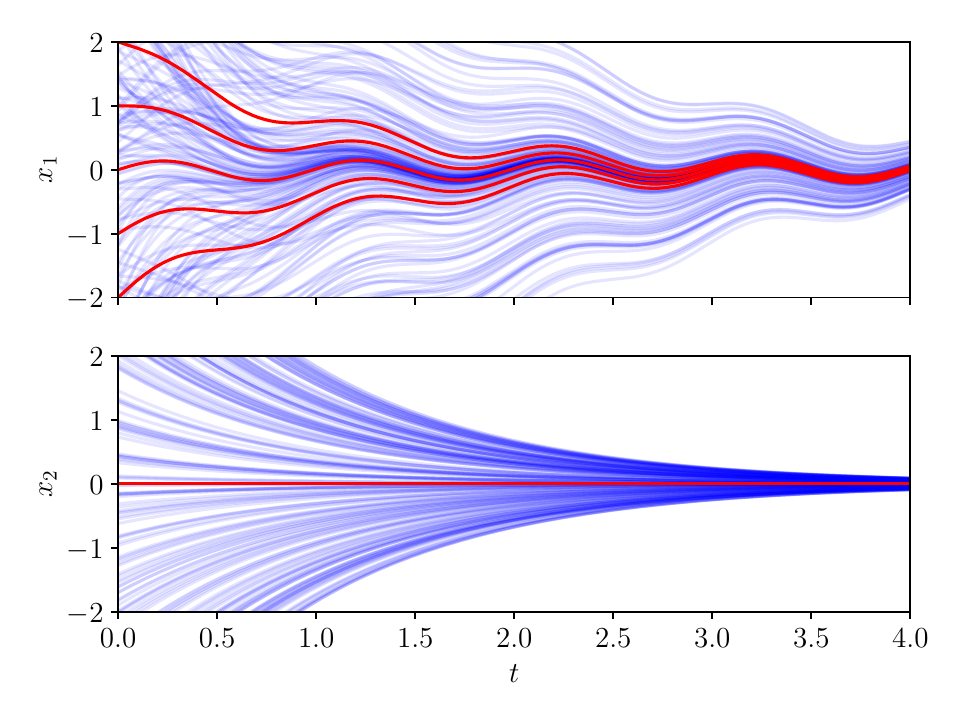}
    \includegraphics[width=0.49\columnwidth]{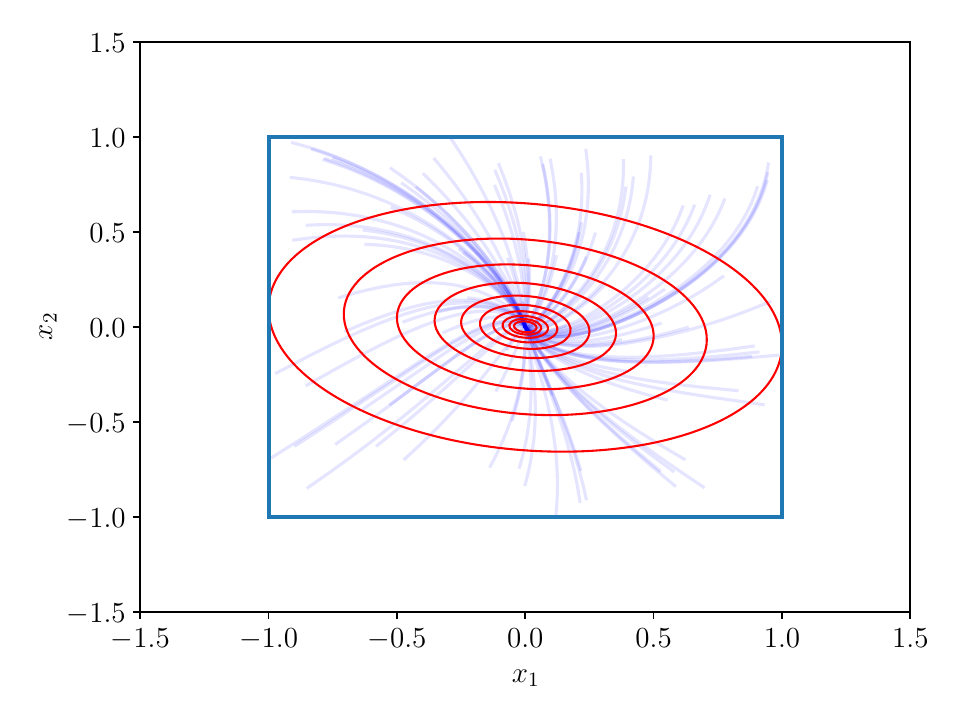}
    \caption{
    \change{\textbf{Left (Example~\ref{ex:taninv})}: Sample trajectories of the system starting in $[-5,5]^2$ are shown in blue. The system entrains with exponential decay rate of $c=\frac{\pi}{4} - 1$ to any known trajectory $x'$ with $x'_2(0) = 0$, which are pictured in red for $x'_1(0)\in\{0,\pm1,\pm2\}$.}
    \textbf{Right (Example~\ref{ex:poly})}: Level sets of the local Lyapunov function $V(x) = x^T P x$, verifying stability with exponential decay rate of \change{$c=-0.45$}, are shown in red. The localizing domain $X = [-1,1]^2$ is pictured in light blue, and sample trajectories starting in $X$ are pictured in blue.
    }
    \label{fig:examples}
\end{figure}

\begin{example} \label{ex:poly}
Consider the time-invariant nonlinear system
\begin{align*}
    \dot{x}=f(x)=\begin{bmatrix}f_1(x)\\f_2(x)\end{bmatrix}=\begin{bmatrix}
    -2x_1 + \frac12 (x_1 + x_2)^2 \\
    -x_1 + \frac12 x_1^2 - 2x_2
\end{bmatrix},
\end{align*}
whose Jacobian matrix is %
\begin{align*}
    Df(x) = \begin{bmatrix}
        -2 + (x_1 + x_2) & 1 + (x_1 + x_2) \\ -1 + x_1 & -2
    \end{bmatrix}.
\end{align*}
We verify the exponential stability of $x' = 0$ and find a region of attraction within $X = [-1,1]^2$. 
The mixed Jacobian is
\begin{align*}
    M_{x'}f(x,s) = \begin{bmatrix}
        -2 + s_1x_1 & 1 + (x_1 + s_2x_2) \\ -1 + s_1x_1 & -2
    \end{bmatrix}.
\end{align*}
First, note that $Df(\smallconc{1}{1}) = \begin{bsmallmatrix} 0 & 3 \\ 0 & -2 \end{bsmallmatrix}$, so $\mu(Df(\smallconc{1}{1})) \geq 0$ for any logarithmic norm $\mu$. 
Noting that $s_1x_1\in[-1,1]$ and $(x_1 + s_2x_2)\in [-2,2]$, $M_{x'}f([-1,1]^2,[0,1]^2)\subseteq\co\{M_1,M_2,M_3,M_4\}$, with
\begin{align*}
    M_1 &= \begin{bsmallmatrix} -3 & 3 \\ -2 & -2 \end{bsmallmatrix}, \quad M_2 = \begin{bsmallmatrix} -3 & -1 \\ -2 & -2 \end{bsmallmatrix}, \\
    M_3 &= \begin{bsmallmatrix} -1 & 3 \\ 0 & -2 \end{bsmallmatrix}, \quad M_4 = \begin{bsmallmatrix} -1 & -1 \\ 0 & -2 \end{bsmallmatrix}.
\end{align*}
Using CVXPY, the following SDP is feasible \change{for $c = -0.45$ and returns $P \approx \begin{bsmallmatrix} 0.425 & 0.093 \\ 0.093 & 0.985 \end{bsmallmatrix}$:}
\begin{align*}
    \change{\max_{P\preceq I}\ \log\det(P)} \ \mathrm{ s.t. } \ & M_i^T P + PM_i \preceq 2cP \ \forall i\in\{1,2,3,4\}.
\end{align*}
Thus, \change{Corollary~\ref{cor:exp_stable} implies that} $x'$ is locally exponentially stable at rate $c = -0.45$, with Lyapunov function $V(x) = |x - x'|_{2,P^{1/2}} = x^T Px$\change{, and any norm ball $\calB^P_\gamma(0)$ for $\gamma>0$ contained in $X$ is a forward invariant region of attraction.}
\end{example}

\subsection[Contraction to Known Trajectories in the l1-Norm]{Contraction to Known Trajectories in the $\ell_1$-Norm}

When comparing to a fixed trajectory $x'$, there are potentially two sources of conservatism when overapproximating the logarithmic norm. 
When using the full set of possible Jacobian matrices $\calJ$, no information regarding the comparison point $x'$ is used, whereas $\calM$ from Theorem~\ref{thm:mixed} uses this information, and in general $Df(x)\notin\calM$.
However, Theorem~\ref{thm:mixed} requires a mean-value theorem application on $n$ different segments. Each column is built from a Jacobian evaluation at a different location, which means that a matrix $M\in\calM$ may not be a Jacobian matrix of the system. 
To summarize, neither $\calM \subseteq \calJ$ nor $\calJ \subseteq \calM$ are generally true, and as a result, neither will necessarily give better contraction estimates.

In the case of the $\ell_1$-norm $|\cdot|_1$, however, the column-wise structure allows us to show that using $\calM$ from Theorem~\ref{thm:main} outperforms the full Jacobian technique that uses $\calJ$.
\begin{thm} \label{thm:ell_1}
    Let $f:\R^n\to\R^n$ be differentiable, 
    $X \subseteq \R^n$ be an interval,
    and let $x'\in X$. Then
    \begin{align*}
        \sup_{x\in X,\,s\in[0,1]^n}\mu_1\left(M_{x'}f(x,s)\right) \leq \sup_{x\in X} \mu_1\left(Df(x)\right).
    \end{align*}
\end{thm}

\newenvironment{talign}
 {\let\displaystyle\textstyle\align}
 {\endalign}
\newenvironment{talign*}
 {\let\displaystyle\textstyle\csname align*\endcsname}
 {\endalign}

\begin{proof}
    The statement follows by swapping the $\sup$ with the $\max$ from the definition of the $\mu_1$ logarithmic norm~\cite[Table 2.1, p. 27]{Bullo:2023}. 
    Using the shorthand $x_{k:l} = x_k, \dots, x_l$,
    \begin{align*}
        &\sup_{x\in X,\, s\in[0,1]^n} \mu_1(M_{x'}f(x,s)) \\
        & = \sup_{x\in X,\, s\in[0,1]^n} \max_{j=1,\dots,n} \big\{M_{x'}f(x,s)_{jj} + \textstyle\sum_{i\neq j} |M_{x'}f(x,s)_{ij}|\big\} \\
        & = \max_{j=1,\dots,n} \sup_{x\in X} \sup_{z_j\in\co(x_j,x'_j)} \big\{\tfrac{\partial f_j}{\partial x_j} (x_{1:j-1}, z_j, x'_{j+1:n}) \\
        &\quad\quad + \textstyle \sum_{i\neq j} |\tfrac{\partial f_i}{\partial x_j} (x_{1:j-1}, z_j, x'_{j+1:n})|\big\}  \\
        & \leq \max_{j=1,\dots,n} \sup_{x\in X} \big\{\tfrac{\partial f_j}{\partial x_j} (x) + \textstyle\sum_{i\neq j} |\tfrac{\partial f_i}{\partial x_j} (x)|\big\} \\
        & = \sup_{x\in X} \max_{j=1,\dots,n} \big\{\tfrac{\partial f_j}{\partial x_j} (x) + \textstyle\sum_{i\neq j} |\frac{\partial f_i}{\partial x_j} (x) |\big\}
        = \displaystyle \sup_{x\in X} \mu_1 (Df(x)),
    \end{align*}
    where the inequality holds since $(x_{1:j-1},z_j,x'_{j+1:n}) \in X$.
\end{proof}

\change{
\begin{rem}
\change{
The bound from Theorem~\ref{thm:ell_1} also holds for any diagonal weighting of the $\ell_1$-norm, \emph{i.e.}, norms defined as $|\Gamma x|_1$ for some diagonal $\Gamma$ with positive entries.
}
\end{rem}
}

\begin{example} \label{ex:l1}
    Consider the nonlinear system 
    \begin{align*}
        \dot{x} = \begin{bmatrix} f_1(x) \\ f_2(x) \end{bmatrix} = \begin{bmatrix} -x_1 \\ -x_2 - x_1x_2 \end{bmatrix}.
    \end{align*}
    We verify the stability of the equilibrium $x' = 0$ and find a region of attraction.
    The Jacobian matrix is
    \begin{align*}
        Df(x) = \begin{bmatrix} -1 & 0 \\ -x_2 & -1 - x_1 \end{bmatrix},
    \end{align*}
    so $\mu_1(Df(x)) = \max(-1 + |x_2|, -1 - x_1)$, while $\mu_1(M_{x'}f(x,s)) = \max(-1 + |x'_2|, -1 - x_1) = \max(-1, -1 - x_1)$.
    Fix $\delta \in(0,1)$. On the set \change{$X = \{x\in \R^2 : x_1\geq -\delta\}$},
    \begin{align*}
        \sup_{x\in X} \mu_1(Df(x)) = +\infty,
    \end{align*}
    and in particular, $\mu_1(Df(x)) < 0$ only on the set $\{x\in X : |x_2| < 1\}$.
    In contrast, since \change{$x_1\geq -\delta$} on $X$,
    \begin{align*}
        \sup_{x\in X,s\in[0,1]^2} &\mu_1(M_{x'}f(x,s)) \leq \max(-1, -1 + \delta) = -1 + \delta.
    \end{align*}
    $X$ is a forward invariant set by Nagumo's theorem \change{($f(-\delta,x_2) = [\delta,\ (-1 + \delta) x_2]^T$)}.
    Thus, applying Corollary~\ref{cor:exp_stable}, the origin is an exponentially stable equilibrium, \change{with decay rate $c= -1 + \delta$ and region of attraction $X$}. 
\end{example}

\begin{figure}
    \centering
    \includegraphics[width=0.5\columnwidth]{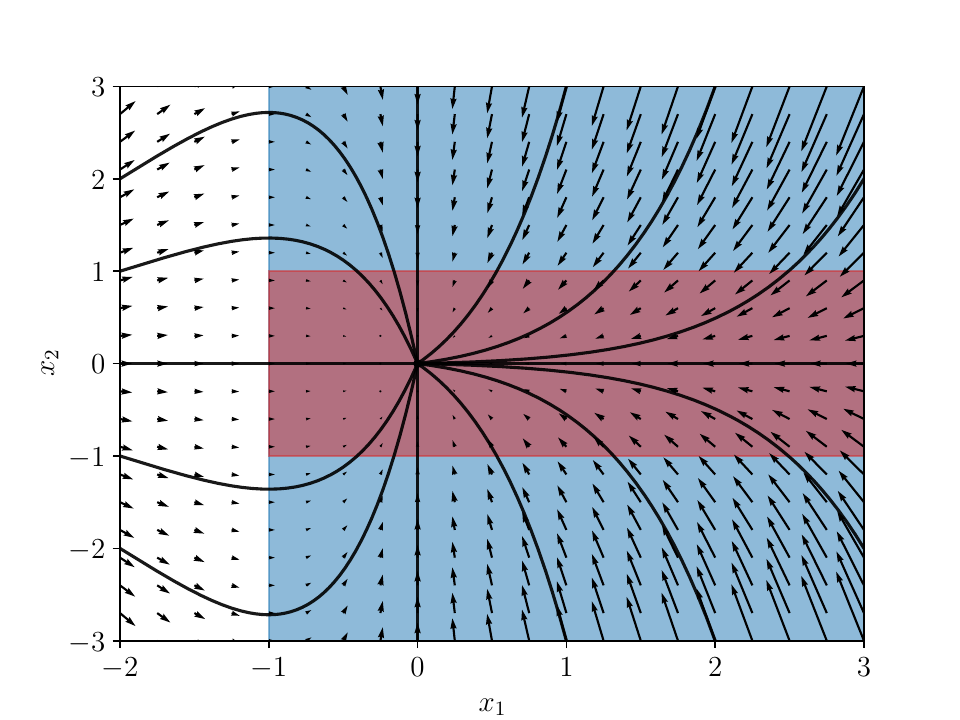}
    \caption{
    The phase portrait from Example~\ref{ex:l1} is pictured. 
    The verified region of attraction $X$ is pictured in blue, and the set where $\mu(Df(x)) < 0$ is pictured in red.
    While the system is not necessarily infinitesimally $\ell_1$-contracting on the blue region $X$, Corollary~\ref{cor:exp_stable} is able to verify that \change{trajectories exponentially decrease their} $\ell_1$-distance to the origin.
    }
    \label{fig:l1}
\end{figure}

\section{Application to Ellipsoidal Reachability} \label{sec:reach}

\change{One} application of contraction analysis is in computing overapproximating reachable sets in nonlinear systems through the following steps:
(i) compute a nominal trajectory $x'(t)$; (ii) bound the logarithmic norm for some region around $x'(t)$; (iii) expand/contract norm balls using this rate~\cite{MaidensArcak:2014}.
Several variations of this simulation-guided approach have been proposed, but to our knowledge, each of the existing approaches bound the logarithmic norm for the entire Jacobian \change{set}.
One approach uses analytically derived bounds for the logarithmic norm and has been used in the study of switched systems~\cite{MaidensArcak:2014b} and in component-wise contraction techniques to improve scalability~\cite{ArcakMaidens:2018}. 
Another approach automatically computes upper bounds for the logarithmic norm using interval bounds for the Jacobian matrix~\cite{FanMitra:2016,FanMitra:2017,Fan:2016}.

From Corollary~\ref{cor:intervalM} and Remark~\ref{rem:interval_over}, we recall that the interval mixed Jacobian $[\calM]$ is never larger than the interval Jacobian $[\calJ]$.
Thus, we immediately obtain an improved automated approach for logarithmic norm-based reachability analysis by replacing the interval Jacobian matrix in~\cite{FanMitra:2017} with the interval mixed Jacobian matrix from~\eqref{eq:intervalMcond}.
In Algorithm~\ref{alg:reach-mjacM}, we provide an implementation of simulation-guided reachability using $[\calM]$.

\begin{algorithm}[tb]
\caption{Ellipsoidal $|\cdot|_{2,P^{1/2}}$-norm reachability using \texttt{mjacM}}
\label{alg:reach-mjacM}
\begin{algorithmic}[1]
    \STATE \textbf{Input:} initial set $\calB^{P_0}_{1}(x'_0)$, step horizon $\Delta t>0$, number of steps $N\in\bbN$
    \STATE $\calR \gets \{0\}\times \calB^{P_0}_1(x'_0)$
    \STATE $x'(t) \gets $ \change{simulated trajectory from initial condition $x'_0$}
    \FOR{$i = \change{0,\dots,N-1}$}
        \STATE $[X_0]\gets x'(i\Delta t) + [-\sqrt{\operatorname{diag}(P^{-1}_{\change{i}})}, \sqrt{\operatorname{diag}(P^{-1}_{\change{i}})}]$ \label{alg:reach-mjacM:iover}
        \STATE Integrate embedding system from $0$ to $\Delta t$, with initial condition $[X_0]$, obtaining $\{[X_t]\}_{t\in [0,\Delta t]}$ \label{alg:reach-mjacM:embsys}
        \STATE \change{Use \texttt{immrax.mjacM} to obtain $\{[\calM_t]\}_{t\in[0,\Delta t]}$ containing $M_{x'(t)} f([X_t],[0,1]^n)$ for every $t\in[0,\Delta t]$}
        \STATE Obtain $\{M^j\}_{j}$ satisfying $\bigcup_{t\in[0,\Delta t]} [\calM_t] \subseteq \co(\{M^j\}_{j})$\label{alg:reach-mjacM:corners}%
        \STATE $P_{\change{i+1}}\gets \change{\argmax}_P \log\det(P)$ s.t. \change{\eqref{eq:reach_demidovich} and \eqref{eq:reach_subset}}
        \STATE $\calR \gets \calR \cup \bigcup_{t\in(0,\Delta t]} \{t + i\Delta t\} \times \calB_{e^{ct}}^{P_{\change{i+1}}} (x'(t))$
        \STATE $P_{i+1} \gets \change{e^{-2c\Delta t}} P_{i+1}$
    \ENDFOR
    \RETURN Reachable tube $\calR$
\end{algorithmic}
\end{algorithm}

Suppose an initial set is specified as the ellipsoid $\calB_1^{P_0}(x'_0)$, and let $x'$ denote the fixed trajectory from initial condition $x'(t_0) = x'_0$.
The first step of Algorithm~\ref{alg:reach-mjacM} is to compute an initial interval over-approximation of the reachable set. 
In the literature~\cite{FanMitra:2016}, this step has been done using \emph{e.g.} Lipschitz bounds of the dynamics. 
Another approach is to use an \emph{inclusion function} to build an \emph{embedding system} whose trajectory bounds the behavior of the original system.
We use the interval analysis toolbox \texttt{immrax}, which automates the construction of this embedding system, and refer to~\cite{immrax} for a description on how interval analysis constructs these embedding systems for interval reachability.

Since the initial set for Algorithm~\ref{alg:reach-mjacM} is an ellipsoid, we  overapproximate the ellipsoid with an interval as the initial condition of the embedding system.
To do this, we note that
\begin{align*}
    \operatornamewithlimits{max}_{x\in\R^n : x^TPx \leq 1} \ell^T x = \sqrt{\ell^T P^{-1}\ell},
\end{align*}
for any $\ell\in\R^n$.
Thus, taking $\ell\in\{\pm e_i\}_{i=1}^n$ yields that the smallest interval containing $\calB_{1}^P(x')$ is $[x' - \sqrt{\operatorname{diag}(P^{-1})}, x' + \sqrt{\operatorname{diag}(P^{-1})}]$, where the $\sqrt{\cdot}$ is elementwise.
Given the initial condition $[X_0]$, the trajectory of the embedding system obtains an initial coarse interval reachable set $[X_t]$, satisfying $x(t)\in[X_t]$. 
The function \texttt{immrax.mjacM} then automatically computes the interval mixed Jacobian matrix $[\calM_t]$ from Corollary~\ref{cor:intervalM} at each time $t$, using automatic differentiation.
Once a finite set of corners $\{M^j\}_j$ satisfying $\bigcup_t[\calM_t] \subseteq \co(\{M^j\}_j)$ is chosen, we search for a feasible solution $(c,P)$ satisfying the following constraints,

\noindent
\begin{minipage}{0.55\columnwidth}
\begin{align}
    M^jP + PM^j \preceq 2cP \ \ \forall j, \tag{$\dagger$} \label{eq:reach_demidovich}
\end{align}
\end{minipage}%
\begin{minipage}{0.45\columnwidth}
\begin{align}
    P \preceq P_0. \tag{$\star$} \label{eq:reach_subset}
\end{align}
\end{minipage}
\vspace{0.25em}

The condition~\eqref{eq:reach_demidovich} is the LMI~\eqref{eq:ell2_LMI} on every corner, implying $\mu_{2,P^{1/2}}(\bigcup_t[\calM_t]) \leq c$ by convexity of $\mu$.
The condition~\eqref{eq:reach_subset} ensures the containment of the radius $1$ norm balls as $\calB_1^P(x'_0) \supseteq \calB^{P_0}_1(x'_0)$.
Since the constraints are convex only for fixed $c$ or $P$, we settle for a line search over $c$, and a maximization of $\log\det(P)$ to minimize the \change{volume} of $\calB_1^P(x'_0)$.
Once a feasible solution $(c,P)$ is found, Theorem~\ref{thm:main} ensures that any trajectory $x$ with initial condition $x(t_0)\in\calB_1^{P}(x'_0)$ satisfies $x(t)\in \calB^P_{e^{c(t - {t_0})}}(x'(t))$ for every $t\geq t_0$.
Rescaling to a radius $1$ ball allows this procedure to iterate to any desired horizon.

\begin{figure}
    \centering
    \includegraphics[width=\columnwidth]{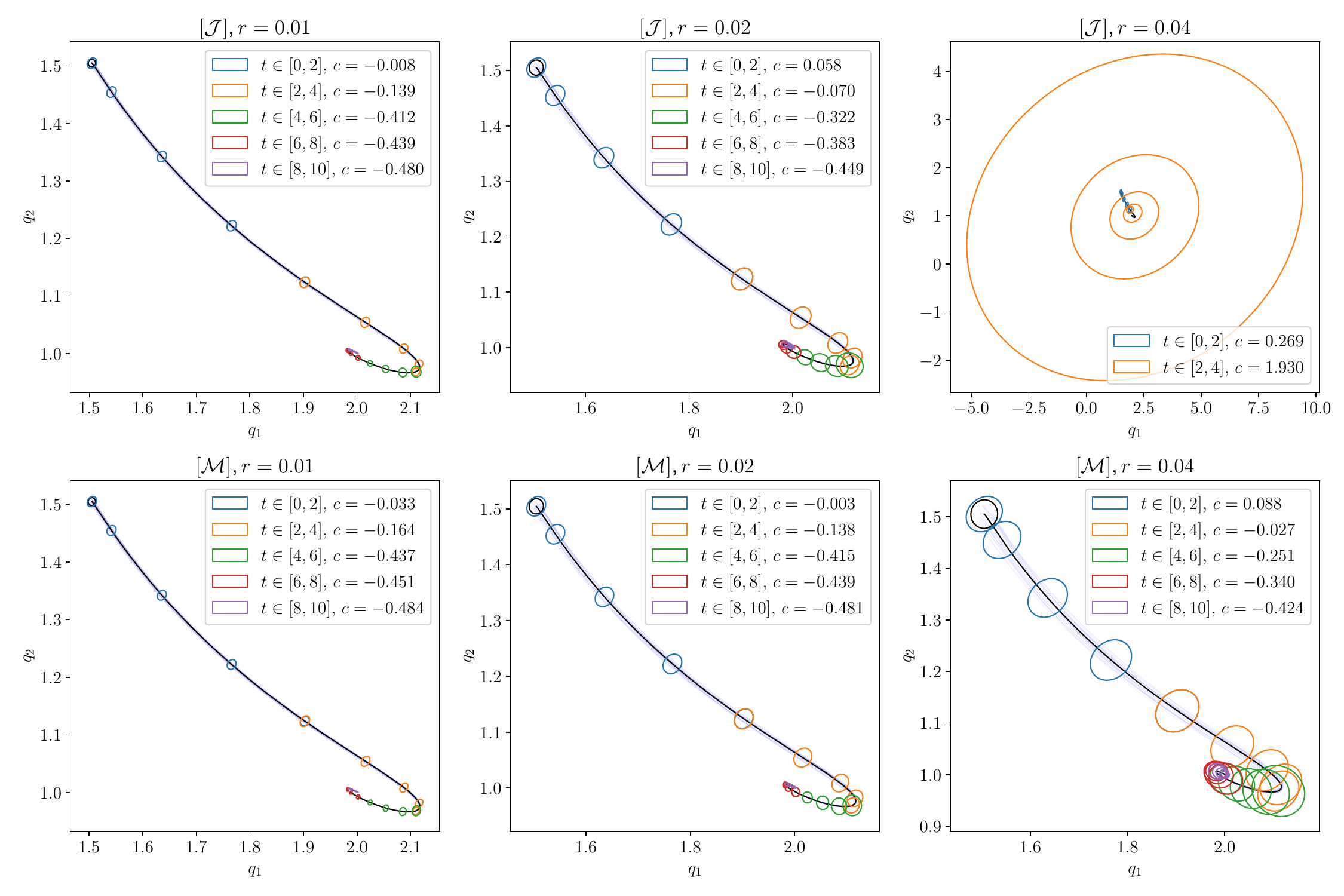}
    \caption{
    The reachable sets from Example~\ref{ex:robot}, comparing $[\calJ]$ (top) and $[\calM]$ (bottom) for varying sizes of initial sets.
    The ellipsoidal projection onto the $q_1$-$q_2$ plane is plotted for $t=0.5k$ (s).
    The known trajectory $x'$ is pictured in black, and several Monte Carlo simulations of the system are pictured in blue.
    $[\calJ]$ fails for the initial set with radius $r = 0.04$.
    }
    \label{fig:robot-arm}
\end{figure}

\begin{example}[Robot arm] \label{ex:robot}
    We compare to the pure Jacobian-based method from~\cite{FanMitra:2016}, using the robot arm model from~\cite{AngeliSontagWang:2000}
    \begin{align*}
        \dot{q}_1 &= z_1, \quad \dot{q}_2 = z_2, \\
        \dot{z}_1 &= \tfrac{1}{mq_2^2 + ML^2/3}(-2mq_2z_1z_2 - k_{d_1}z_1 + k_{p_1}(u_1 - q_1)), \\
        \dot{z}_2 &= q_2 z_1^2 + \tfrac1m (- k_{d_2} z_2 + k_{p_2} (u_2 - q_2)),
    \end{align*}
    with $u_1 = 2$, $u_2 = 1$, $m=M=1$, $L=\sqrt{3}$, $k_{p_1} = 2$, $k_{p_2} = 1$, $k_{d_1} = 2$, $k_{d_2} = 1$.

    We run Algorithm~\ref{alg:reach-mjacM} with hyperparameters $\Delta t = 2$, $N = 5$, with varying initial sets $\calB_r^P(x'_0)$ with $P$ and $x'_0$ from~\cite{FanMitra:2016}, and $r \in \{0.01,0.02,0.04\}$.
    We order the state variables into the following permutation $[z_1\ z_2\ q_1\ q_2]^T$ (see Remark~\ref{rem:permuting}).
    For the embedding system integration step at line~\eqref{alg:reach-mjacM:embsys}, we use Euler integration with a step size of $h=0.001$.
    At line~\eqref{alg:reach-mjacM:corners}, we take the interval union of the $\frac{\Delta t}{h} = 2000$ interval matrices computed during the integration step, \emph{i.e.}, the smallest interval matrix $[\calM]$ containing all $2000$ of these matrices, and sparsely extract the $64$ corners of $[\calM]$.
    There are only $64$ corners since the Jacobian matrix has only $6$ nonconstant elements.
\end{example}

\section{Conclusion}

In this note, we constructed a new linear differential inclusion bounding the behavior of nonlinear systems by traversing an element-wise path from $x'$ to $x$ rather than the traditional straight-line path.
When comparing to a known trajectory $x'$, this approach provides potential improvement compared to bounding the full Jacobian matrix as in previous approaches. 
For instance, we demonstrated computational improvement for interval-based algorithms bounding the logarithmic norm for reachability analysis. 

\vspace{-1em}

\bibliographystyle{ieeetr}
\bibliography{ldi.bib}

\end{document}